\documentclass[a4paper, hidelinks, 11pt]{article}
\usepackage{graphicx} 
\usepackage{amsmath, relsize} 
\usepackage{todonotes} 
\usepackage{amsthm} 
\usepackage{amsfonts}
\usepackage{amssymb}
\usepackage{url}
\usepackage{mathrsfs}
\usepackage{amsmath}
\usepackage{amssymb}
\usepackage{amsthm}
\usepackage{bbm}
\usepackage{comment}
\usepackage{subcaption}
\usepackage{tocloft}

\usepackage{enumerate}
\usepackage{setspace}
\usepackage{graphicx,color,hyperref}
\usepackage[margin=0.8in]{geometry}
\parskip \medskipamount
\newtheorem{theorem}{Theorem}[section]

\theoremstyle{definition}

\newcommand{\G}     {\mathbb{G}} 
\newcommand{\R}     {\mathbb{R}} 
\newcommand{\Z}     {\mathbb{Z}} 
\newcommand{\N}     {\mathbb{N}} 
\renewcommand{\P}   {\mathbb{P}}

\newcommand{\T}     {\mathbb{T}}

\numberwithin{equation}{section}

\title{Macroscopic  loops \\ in the $3d$  double-dimer model}
\author{Alexandra Quitmann, Lorenzo Taggi}

\begin{document}
\maketitle

\begin{abstract}
The double dimer model is defined as the superposition of two independent uniformly distributed dimer covers of a graph.  Its configurations can be viewed as disjoint collections of self-avoiding loops. 
Our first result is that  in $\mathbb{Z}^d$, $d>2$,
the loops in the double dimer model are macroscopic.
These are shown to behave 
qualitatively differently than in two dimensions.  In particular,
we show that, given two distant points of a large box, with uniformly positive probability there exists a loop visiting both points.
Our second result involves the  monomer double-dimer model, namely the double-dimer model in the presence of a density of monomers. These are vertices  which are not allowed to be touched by any loop.
This model depends on a parameter, the monomer activity, which 
controls the density of monomers. 
It is known from \cite{BetzTaggi,T} that a finite critical threshold of the monomer activity exists,  below which a self-avoiding walk forced through the system  is macroscopic. 
Our paper shows that,  when $d >2$, such a critical threshold is strictly positive.
In other words,   the  self-avoiding walk is macroscopic even in the presence of a positive density of monomers. 
\end{abstract}

\section{Introduction}
Dimer covers are  perfect matchings
of a graph, namely spanning sub-graphs such that every vertex has degree  one. 
These  mathematical objects attract 
interest from a wide range of perspectives, which include combinatorics,   probability,  statistical mechanics, and algorithm complexity studies. 
Our paper considers a random walk loop soup consisting  in the superposition of 
two independent uniformly distributed dimer covers. 
More precisely, we consider the double dimer model, which consists in the superposition
of two independent uniformly distributed dimer covers, 
and the monomer double-dimer model,  which corresponds to the double dimer model
in the presence of a density of monomers (a monomer is a vertex
which is not allowed to be touched by any loop).
These models are related to the loop O(N) model and to other random walk loop soups
which received increased interest in the last few years, 
see for example \cite{BenassiUeltschi, BetzUeltschi, DuminilCopinParafermionic,  L-T-exponential, PeledSpinka,QuitmannTaggi}
for some references.

Our first result involves the double dimer model.
The planar case was studied in \cite{Dubedat2, Kenyon5},
in which conformal invariance properties of the scaling limit were proved.
Relying on Kasteleyn's theorem \cite{Kasteleyn, TemperleyFisher}, the methods of these papers  do not apply  to higher dimension. Our first main theorem shows that the loops of the double dimer model in $\mathbb{Z}^d$, $d>2$, are macroscopic.  More precisely, 
we consider the double dimer model on a torus of $\mathbb{Z}^d$ of $L^d$ sites
and show that 
 \textbf{(1):} the expected length of each loop is of order $O(L^d)$ \textbf{(2):} 
 given two vertices
on the Cartesian axis having distance of order $O(L)$,  with uniformly positive probability a loop connects both.  Contrary to our high dimensional case, in the planar case this probability converges to zero as $L \rightarrow \infty$ \cite{Kenyon5}.

Secondly,  we introduce the monomer double-dimer model.
The model depends on a parameter, the \textit{monomer activity},  which controls the density of monomers, and reduces to the double dimer model when the monomer activity is zero. 
We consider a version of this model where one of the loops is forced to be `open' and then corresponds to a self-avoiding walk starting from the origin, which is allowed to end at an arbitrary vertex of the box.   It is known from \cite{BetzTaggi} that the length of the self-avoiding walk admits uniformly bounded exponential moments as the monomer activity is very large. 
It is known from \cite{T} that the self-avoiding walk is `long' if the monomer activity is zero and $d>2$, i.e, the distance between its two end-points grows as the diameter of the box. 
Our second main result, Theorem \ref{theo:theo2}, 
is that the  self-avoiding walk keeps being long even if the  monomer activity is strictly positive.  In other words,  the phase transition
in this model is non-trivial and occurs at   a strictly positive and finite threshold
of the monomer activity.  
Our results also hold for an extension of this model
in which the measure depends on a parameter $N$ rewarding the total number of loops in the system.

Our results are an extension of the method developed in  \cite{T}, 
in which the reflection positivity  technique
has been reformulated in the framework of random walk loop soups. 

\section{Definitions and results}
Consider a finite undirected graph $G = (V,E)$. 
A dimer configuration is a spanning sub-graph of $G$ such that every vertex has degree  one. 
We let $\mathcal{D}_G$ be the set of all dimer configurations in $G$. 
Given a set $A \subset V$, we let $G_A$ be the subgraph of $G$ with vertex set
$V \setminus A$ and with edge-set consisting of all the edges  in $E$ which 
do not touch any vertex in $A$.
We let $\mathcal{D}_G(A)$ be the set of dimer configurations in $G_A$. 
We let $\mathbb{T}_L =(V_L, E_L)$ be a graph corresponding to a torus with vertex set $V_L = (-\frac{L}{2},\frac{L}{2}]^d \cap \mathbb{Z}^d$ and with edges connecting nearest neighbour and boundary vertices.
We let $o \in V_L$ be the origin.

\paragraph{The double-dimer model.}
The double-dimer model is sampled by superimposing two independent dimer covers with uniform distribution on the set of possible dimer covers. 
Each realisation of the model can then be viewed as a collection of disjoint self-avoiding loops.

We let $P_{L}$ be the counting measure on the set $\mathcal{D}_{\mathbb{T}_L} \times \mathcal{D}_{\mathbb{T}_L}$ normalized by the total number of double dimer coverings and we denote by $E_L$ the corresponding expectation.  We note that to each pair  $(d_1,d_2) \in \mathcal{D}_{\mathbb{T}_L} \times \mathcal{D}_{\mathbb{T}_L}$ there corresponds a unique set of disjoint loops (some of these loops may correspond to the superposition of two dimers on the same edge).
Moreover, we note that to each such set of loops there correspond several double dimer configurations. 
We denote  by $\{o \leftrightarrow x\}$ the set of double dimer covers such that both $o, x \in V_L$ belong to the same loop. 
Further, we denote by $L_o =  L_o(d_1, d_2)$ the loop that contains the origin and by $|L_o|$ we denote its length, namely,
$
|L_o| := \sum_{x \in V_L} \mathbbm{1}_{\{x \in L_o\}}.
$

Our first theorem states that the expected length of the loop that contains the origin is of the same order of magnitude as the volume of the box and that the probability that a loop connects two vertices whose distance is of order $O(L)$ is uniformly positive. 
Moreover, our theorem also provides an upper bound on the expected loop length.
To state the theorem, let 
 $N_+=\sum_{n>0} \mathbbm{1}_{\{S_n=0\}}$ be the number of returns to the origin of an independent simple random walk $S_n$ in $\Z^d$ starting at the origin. Denoting its expectation by $E^d$, we set $r_d:=E^d[N_+]$. 
 \begin{theorem} \label{theo:theo1}
Suppose that $d>2$. Then
\begin{equation} \label{eq:theo1}
 \bigg(\frac{1}{2d} \Big(  1 - \frac{r_d}{2}   \Big)\bigg)^2  \, \leq \, \liminf\limits_{ \substack{L \rightarrow \infty \\ \mbox{ \footnotesize $L$ even } } } \, \, \frac{1}{| V_L |}  \, E_L\big[|L_o|\big] 
\, \leq \,  \frac{1}{2d} \, \Big(2-\frac{1}{2d}\Big).
\end{equation}
Moreover, for any $\varphi \in (0, \frac{1}{2d} (  1 - \frac{r_d}{2}   ))$, there exists $\varepsilon=\varepsilon(\rho,d) \in (0,\frac{1}{2})$ such that for any $L \in 2\N$ large enough and any odd integer $n \in (0, \varepsilon \, L)$,
\begin{equation} \label{eq:eq2theo1}
P_L(o \leftrightarrow n \, \boldsymbol{e}_1) \geq \varphi^2.
\end{equation}
\end{theorem}
An exact computation made by Watson \cite{Watson} shows that $0.51 < r_3 < 0.52$, and the Rayleigh monotonicity principle \cite{Lyons} implies that $r_d$ is non-increasing in $d$. This implies that the expression in the square on the left-hand side of (\ref{eq:theo1}) is uniformly positive for any $d >2$.
We refer to \cite{Joyce} for numerical estimates of $r_d$ for $d \geq 3$. 

 Our Theorem \ref{theo:theo1}  extends  \cite[Theorem 1.1]{QuitmannTaggi},
 which states  the occurrence of macroscopic loops for a very general random walk loop soup.
 That theorem, however, does not cover the double dimer model, since it holds
 only for random walk loop soups in which the vertices
 of the graph are allowed to be visited sufficiently many times by the loops.

\paragraph{The monomer double-dimer model.}
The monomer double-dimer model generalizes the double dimer model. The generalisation consists in  allowing the presence of a density of monomers.  These are controlled through an external parameter, the monomer activity.  
When the monomer activity is zero, the model reduces to the double-dimer model.

The configuration space of the monomer double-dimer model is denoted by $\Omega$ and it corresponds to the set of triplets $\omega = (M, d_1, d_2)$ such that
$M \subset V_L$ and $(d_1, d_2) \in \mathcal{D}_{\mathbb{T}_L}(M) \times \mathcal{D}_{\mathbb{T}_L}(M)$. We refer to the first element of the triplet $\omega$ as a set of monomers. 
We let $\mathcal{M} : \Omega \mapsto V_L$ be the random variable
corresponding to the set of monomers,  i.e,   
$\mathcal{M}(\omega) := M$ for each $\omega = (M, d_1, d_2) \in \Omega$.
 \begin{figure}[t]
\centering
\includegraphics[scale=1.1]{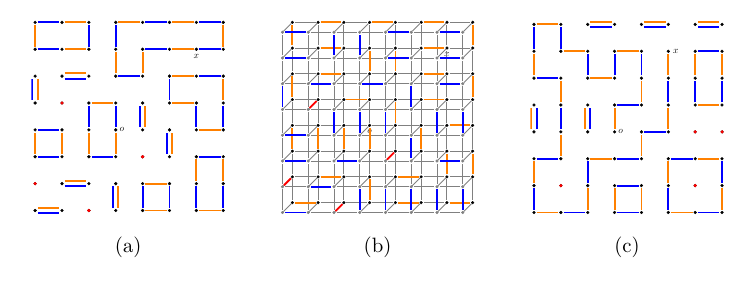}
\caption{(a) A configuration of the monomer double-dimer model.  (b) 
A dimer configuration on the duplicated torus. 
The red, orange and blue dimers represent the vertical dimers and the dimers of the lower and upper torus, respectively.  Projecting the  upper and lower torus onto the same torus leads to the configuration on the left. (c) A realisation $\omega =(M, d_1, d_2) \in \Omega(o,x)$. The set of monomers $M$ is represented by red circles. The dimers of $d_1$ and $d_2$ are coloured blue and orange, respectively. 
} 
    \label{fig:monomerdoubledimer}
\end{figure}
As one can see on the left of Figure  \ref{fig:monomerdoubledimer},
any such configuration can be viewed as a configuration of disjoint self-avoiding loops. As for the double dimer model, some of such loops may correspond to the superposition of two dimers
on the same edge.

Alternatively,  the configurations 
 can be viewed as dimer configurations
on a graph corresponding to two copies of a torus, 
with an edge connecting each vertex of one of the two tori to the corresponding
vertex of the other torus.
See also Figure \ref{fig:monomerdoubledimer}.
We refer to such a graph as \textit{duplicated torus} and to the edges connecting 
the two tori as  \textit{vertical edges.}
The representation is such that each monomer in the first representation corresponds to a dimer on a vertical edge in the second representation.
In this representation, the monomer activity then rewards the number of dimers on vertical edges.

A further alternative natural description of this model is to consider its configurations as permutations of the vertices of the graph such that each vertex is mapped either to itself or to a nearest neighbour, see \cite{T} for more details.

To state our second main result we need to introduce the  \textit{two point function}
of the  monomer double dimer model,  which is defined by 
placing two monomers at two vertices.
More precisely, 
given any two pair of vertices $x, y \in V_L$,  we define
$\Omega(x,y)$ as the set of triplets $\omega =(M, d_1, d_2)$ such that
$M \subset V_L \setminus \{x,y\}$, 
$d_1 \in \mathcal{D}_{ \mathbb{T}_L  }(M \cup \{x, y\})$, 
and $d_2 \in \mathcal{D}_{ \mathbb{T}_L  }(M)$
if $|x-y|$ is odd and 
$d_1 \in \mathcal{D}_{ \mathbb{T}_L  }(M \cup \{x\})$, 
$d_2 \in \mathcal{D}_{ \mathbb{T}_L  }(M \cup \{y\})$
if $|x-y|$ is even.  
As one can see on the right of Figure \ref{fig:monomerdoubledimer}, each element $\omega \in 
\Omega(x,y)$ can be viewed as a collection of disjoint self-avoiding 
loops with a self-avoiding walk starting at $x$ and ending at $y$.
For any $\omega \in \Omega(x,y)$, we let $\mathcal{L}(\omega)$ be the number
of loops in $\omega$ (we consider a loop also the object corresponding
to the superposition of two dimers on the same edge). 
We let $N \geq 0$, define  the partition functions,
\[
\mathbb{Z}_{L, N, \rho} (x,y) := \sum\limits_{\omega \in \Omega(x,y) } \rho^{|
\mathcal{M}(\omega)|} \, {\big (\frac{N}{2} \big )}^{  \mathcal{L}(\omega) },
\]
\[
\mathbb{Z}_{L, N, \rho} := \sum\limits_{\omega \in \Omega  } \rho^{|
\mathcal{M}(\omega)|} \,  {\big (\frac{N}{2} \big )}^{  \mathcal{L}(\omega) },
\]
and introduce the \textit{two-point function},
\begin{equation*}\label{eq:monomercorrelationdd}
\mathbb{G}_{L, N, \rho} (x,y) := \frac{\mathbb{Z}_{L, N,  \rho}(x,y)}{ \mathbb{Z}_{L, N, \rho} }.
\end{equation*}
Recalling the description of  the monomer-double dimer model as a dimer model on the duplicated torus, the two-point function can then be viewed as   the ratio  between the 
weight of all dimer configurations on the duplicated torus with  two monomers displaced at two vertices 
of such a graph and the weight of all configurations with no such monomers.  
When $\rho=0$ and $N=2$,  our two point function reduces to the monomer-monomer correlation of the \textit{dimer model} on the torus, 
\[
\mathbb{G}_{L, 2, 0} (x,y)  =  \frac{{|\mathcal{D}_{\mathbb{T}_L}(\{x,y\} })|}{|\mathcal{D}_{\mathbb{T}_L}|  },
\]
which is known to decay polynomially with the distance between the two
monomers on $\mathbb{Z}^2$  \cite{Dubedat2}
and to be uniformly positive on $\mathbb{Z}^d$ with  $d>2$ \cite{T}.
Our second main theorem states that the two-point function stays uniformly positive even for strictly positive values of the monomer activity (and even for integer values of $N$ different from two, which are not too large).
Since it is known from  \cite{BetzTaggi} that the two-point function decays exponentially fast when the monomer activity is large enough,
our result implies the occurrence of a phase transition at a strictly positive threshold
for the monomer activity when $d>2$.

An alternative formulation of our result involves the length of a self-avoiding walk forced through the system of loops.
Define the set $\Omega^w : = \cup_{x \in V_L} \Omega(o,x)$,
whose elements can be viewed as systems of mutually-disjoint self-avoiding loops
with a self-avoiding walk starting at the origin and ending at an arbitrary vertex of the box. 
Define 
a probability measure on this set, 
\[
\forall \omega \in \Omega^w \quad
\mathbb{P}_{L, N, \rho}(\omega)  := 
\frac{\rho^{|
\mathcal{M}(\omega)|} \, \big (\frac{N}{2} \big )^{  \mathcal{L}(\omega) }}{   \sum\limits_{\omega \in \Omega^w } \rho^{|
\mathcal{M}(\omega)|} \, \big (\frac{N}{2} \big )^{  \mathcal{L}(\omega) }   }.
\]
For any $\omega \in \Omega^w$,  let  $X = X(\omega)$ be the random end-point of the self-avoiding walk
which is not the origin
(if $\omega \in \Omega(o,o)$, then the self-avoiding walk is degenerate 
and consists of a single monomer at $o$, in this case we set $X= o$).
Our second main theorem, Theorem \ref{theo:theo2} below, states that,
if the monomer activity is small enough (possibly positive) and $N$ is a not too large integer, then the distance between the two end-points of the self-avoiding walk 
is of the same order of magnitude as the diameter of the box. 
Contrary to this, it is known from \cite{BetzTaggi, TaggiShifted} that 
the length of the walk is of order  $O(1)$ in the limit of large $L$ if the monomer activity is large enough. 
 \begin{theorem}\label{theo:theo2}
Suppose that $d  > 2$, that $N$ is an integer in $\big(0, \frac{2 \, (3\pi-4)}{\pi \, r_d}\big)$ and that $\rho \in \big[0, 1- \frac{\pi \, r_d \, N}{2 \, (3\pi-4)}\big)$.
Then,
\begin{equation}\label{eq:permpositivityaverage}
 \liminf\limits_{ \substack{L \rightarrow \infty \\ \mbox{ \footnotesize $L$ even } } } \, \, \frac{1}{| V_L |}  \sum_{x \in V_L}  \mathbb{G}_{L, N, \rho}(o,x)  \, 
 \geq \, \frac{1}{4d} \Big(  \frac{3\pi-4}{\pi N} \, (1-\rho)- \frac{r_d}{2}   \Big) =: C.
\end{equation}
Moreover, for any  $\varphi \in (0,  \frac{C}{2})$, there exists $\varepsilon >0$ such that for any $L \in 2\N$ large enough, any odd integer $n \in (-\varepsilon \, L, \varepsilon \, L)$ and any $i \in \{1,\dots,d\}$, 
\begin{equation} \label{eq:eq2theo2}
 \mathbb{G}_{L, N, \rho}(o,n \, \boldsymbol{e}_i) \geq \varphi.
\end{equation}
Finally,  under the same assumptions,  there exists $c>0$ such that 
for any $L \in 2 \mathbb{N}$,
\begin{equation} \label{eq:eq3theo2}
\mathbb{P}_{L, N, \rho} \big (  \, | X  |_1  \, > \,  c \, L  \,  \big )   > c.
\end{equation}
\end{theorem}
For $d=3$,  the right-hand side of \eqref{eq:permpositivityaverage} is strictly positive for any $N \in \{1,\dots,6\}$.
Our theorem extends Theorems 2.2 and 2.3 in \cite{T},  since 
our lower bounds (\ref{eq:permpositivityaverage}),
 (\ref{eq:eq2theo2}),  and (\ref{eq:eq3theo2})
 hold even for strictly positive values of the monomer activity.

\section{Proof of the theorems}
Our Theorem \ref{theo:theo1} follows from   \cite[Theorem 1]{T}
and   from an argument from C.  Kenyon \textit{et al.} in \cite{KenyonC}.
 Theorem 1 in \cite{T} shows that the monomer-monomer 
correlation of the dimer model  is uniformly positive. 
Combining this result with an argument from  \cite{KenyonC}, we
deduce that the expected loop length in the double dimer model is uniformly positive. By  using a general monotonicity property for   two-point functions
 which was proved in \cite{L-T-quantum}, we also deduce that the probability of existence of a loop connecting two given distant vertices on the Cartesian axis is also uniformly positive.

The starting point of the proof of our Theorem \ref{theo:theo2} is the lower bound
for the Ces\`{a}ro sum of the two-point function, equation  (\ref{eq:infrared}) below,
which was derived in \cite{T}
using the reflection positivity method. 
This inequality can be viewed as a version of the so-called \textit{Infrared Bound},
which is used in the framework of  spin systems \cite{Frohlich0}. 
For spin systems,  the uniform positivity of the Ces\`{a}ro sum 
follows immediately from the Infrared bound, since the term  $o-o$ of the two-point function
stays uniformly positive in the external parameter (the so-called inverse temperature) and the bound
gets better and better as the inverse temperature gets larger. 
This is not the case for monomer double-dimer model, in which the term $o-o$ of the two-point function
vanishes as the monomer activity goes to zero. 
Hence, we deal with an alternative version of the Infrared bound,
represented by equation  (\ref{eq:infrared}) below.
To obtain non-trivial information from such an inequality  one needs to show that the right-hand side of 
the inequality is uniformly positive.  For that, one needs to provide a uniform lower bound to the
third term in the right-hand side of  (\ref{eq:infrared}),
which involves a sum of two-point functions at the even sites
of the torus multiplied by positive and negative coefficients,  and compare it with the term $\mathbb{G}_{L, N, \rho}(o, \boldsymbol{e}_1)$, which is shown to be uniformly positive and non-decreasing with the monomer activity.
To control such a sum, the assumption $\rho=0$   was made in  \cite{T}.
Under this assumption  the  even two-point function is just zero, hence such a  sum vanishes. 
Here we use simple analytic methods to provide  a general  lower bound 
to this term which holds for positive values of the monomer activity and is finite (but negative) uniformly in the size of the torus. 
Combined with the uniform positivity of the term $\mathbb{G}_{L, N, \rho}(o, \boldsymbol{e}_1)$,
this allows us to deduce that the Ces\`{a}ro sum in the left-hand side is uniformly positive
even for positive values of the monomer activity. 

Our analysis can then be applied to other models 
for which an Infrared bound is available but, contrary to the classical case  \cite{Frohlich0}, the term
 $o-o$ of the two-point function is too  small
 to deduce long range order,
 while the term  $o-\boldsymbol{e_1}$
 is large.

\subsection{Proof of Theorem \ref{theo:theo1}}
We introduce the subsets of even and odd vertices of $V_L$, namely,
\begin{equation*}
V_L^e := \{ x \in V_L: \, d(o,x) \in 2\N_0\}, \qquad V_L^o := \{ x \in V_L: \, d(o,x) \in 2\N_0+1\},
\end{equation*}
where $d(o,x)$ is the graph distance in $\mathbb{T}_L$. For lighter notation, we will omit sub-scripts where appropriate.

The starting point of our proof of Theorem \ref{theo:theo1} is the following uniform positivity result for the monomer-monomer correlation of the dimer model. 

\begin{theorem}[{\cite[Theorem 2.1]{T}}] \label{theo:theoLorenzo}
Suppose that $d>2$. Then
\begin{equation} \label{eq:dimeruniform}
 \liminf\limits_{ \substack{L \rightarrow \infty \\ \mbox{ \footnotesize $L$ even } } } \, \, \frac{1}{| V_L^o |}  \sum_{x \in V_L^o} \frac{|\mathcal{D}(\{o,x\})|}{|\mathcal{D}(\emptyset)|} \geq \frac{1}{2d} \, \Big(1-\frac{r_d}{2}\Big).
\end{equation}
Moreover, for any $\varphi \in (0,\frac{1}{2d} (1-\frac{r_d}{2}))$, there exists a constant $c = c(\varphi, d) \in (0,\frac{1}{2})$ such that for any $L \in 2\N$ large enough and any odd integer $n \in (0, c \, L)$,
\begin{equation} \label{eq:dimeruniform2}
\frac{|\mathcal{D}(o, n \, \boldsymbol{e}_1)|}{|\mathcal{D}(\emptyset)|} \geq \varphi.
\end{equation}
\end{theorem}
In our proof we will first bound the number of double dimer covers containing a loop connecting $o$ and $x$ from below by $|\mathcal{D}(\{o,x\})|^2$ and then apply Theorem \ref{theo:theoLorenzo} to deduce uniform positivity of the expected loop length in the double dimer model.

For lighter notation, we will write $\mathcal{D}(A)^2$ for the set $\mathcal{D}(A) \times \mathcal{D}(A)$ for any $A \subset V_L$.

\begin{proof}[Proof of Theorem \ref{theo:theo1}]
To begin, we show that for any $x \in V_L^o$, it holds that
\begin{equation} \label{eq:monomermonomerloop}
| \mathcal{D}(\{o,x\})|^2 \leq |\mathcal{D}(\emptyset)^2: \, o \leftrightarrow x|.
\end{equation}

The upper bound \eqref{eq:monomermonomerloop} follows from the proof of \cite[Theorem 2]{KenyonC}. For comprehensiveness, we repeat the main argument given there. Note that if $x$ is adjacent to the origin, then \eqref{eq:monomermonomerloop} trivially holds true since for any pair of configurations $(d_1, d_2) \in \mathcal{D}(\{o,x\})^2 $ we can add precisely one dimer on the edge $\{o,x\}$ in each of the two configurations $d_1$ and $d_2$, thus obtaining a configuration in $\mathcal{D}(\emptyset)^2 $ which contains a loop consisting of two dimers on $\{o,x\}$. 
Suppose now that $x$ is not adjacent to the origin. We introduce the map $\phi: \mathcal{D}(\{o,x\}) \times \mathcal{D}(\{\boldsymbol{e}_1,x+\boldsymbol{e}_1\}) \to \{ \mathcal{D}(\emptyset)^2 : \, o \leftrightarrow x \}$, which for each configuration $(d_1,d_2) \in \mathcal{D}(\{o,x\}) \times \mathcal{D}(\{\boldsymbol{e}_1,x+\boldsymbol{e}_1\})$ acts by 
\begin{itemize} 
\item[(i)] colouring the dimers in $d_1$ orange and the dimers in $d_2$ blue,
\item[(ii)] superimposing both (coloured) configurations,
\item[(iii)] adding an orange dimer on $\{o, \boldsymbol{e}_1\}$ and on $\{x,x+\boldsymbol{e}_1\}$, by then switching the colours of the dimers along the path from $x+\boldsymbol{e}_1$ to $\boldsymbol{e}_1$ which does not touch $x$ from blue to orange and orange to blue, respectively, 
\item[(iv)] defining a new configuration $(d_1^\prime,d_2^\prime) \in \mathcal{D}(\emptyset)^2$, where $d_1^\prime$ consists of all orange and $d_2^\prime$ consists of all blue dimers.
\end{itemize}
Note that the configuration $(d_1^\prime,d_2^\prime)$ contains a loop which visits amongst others the set $\{o,\boldsymbol{e}_1,x,x+\boldsymbol{e}_1\}$. For an illustration see also \cite[Figure 2]{KenyonC}. The map $\phi$ is well-defined and an injection as shown in \cite{KenyonC}. Since $|\mathcal{D}(\{\boldsymbol{e}_1,x+\boldsymbol{e}_1\})| = |\mathcal{D}(\{o,x\})|$, this proves \eqref{eq:monomermonomerloop}.

Using the Cauchy-Schwarz inequality, we then have that
\begin{equation} \label{eq:CauchySchwarzapp}
\frac{1}{|V_L^o|} \, \sum_{x \in V_L^o} \frac{|\mathcal{D}(\{o,x\})|}{|\mathcal{D}(\emptyset)|} \leq \bigg(\frac{1}{|V_L^o|} \, \sum_{x \in V_L^o} \frac{|\mathcal{D}(\{o,x\})|^2}{|\mathcal{D}(\emptyset)|^2}\bigg)^\frac{1}{2} \leq \bigg(\frac{1}{|V_L^o|} \, \sum_{x \in V_L^o} P_L(o \leftrightarrow x) \bigg)^{\frac{1}{2}},
\end{equation}
where we used \eqref{eq:monomermonomerloop} in the last step. Since the number of even and odd vertices in each loop coincide we further have that
\begin{equation} \label{eq:expeclength}
\frac{1}{|V_L|} \, E_L\big[|L_o|\big] = \frac{1}{|V_L^o|} \, \sum_{x \in V_L^o} E_L\big[\mathbbm{1}_{\{x \in L_o\}}\big] = \frac{1}{|V_L^o|} \, \sum_{x \in V_L^o} P_L(o \leftrightarrow x).
\end{equation}
From \eqref{eq:dimeruniform}, \eqref{eq:CauchySchwarzapp}, \eqref{eq:expeclength}, and from the fact that $\frac{1}{2d} \, \big(1-\frac{r_d}{2}\big)>0$,  we deduce the lower bound in \eqref{eq:theo1}. The upper bound follows from \eqref{eq:expeclength} and from the site-monotonicity property,
\begin{equation} \label{eq:sitemonotonicityappl}
P_L(o \leftrightarrow x) \leq P_L(o \leftrightarrow \boldsymbol{e}_1) = \frac{1}{2d} \, \Big(2-\frac{1}{2d}\Big)
\end{equation}
for all $x \in V_L$. In \eqref{eq:sitemonotonicityappl} we applied \cite[Theorem 2.1]{L-T-quantum} noting that the quantity $P_L(o \leftrightarrow x)$ can be expressed in the language of the \textit{random path model} for specific choices of the parameters (see e.g. \cite{L-T-first, L-T-quantum, QuitmannTaggi}). The rightmost term in \eqref{eq:sitemonotonicityappl} corresponds to the probability that at least one of the two dimer configurations has a dimer on $\{o,\boldsymbol{e}_1\}$. 
From \eqref{eq:monomermonomerloop} and \eqref{eq:dimeruniform2} we further deduce \eqref{eq:eq2theo1}. 
This concludes the proof of the theorem.
\end{proof}

\subsection{Proof of Theorem \ref{theo:theo2}}
In this section we prove Theorem \ref{theo:theo2}. We will use the notation 
\begin{equation*}
\G_{L,N,\rho}(x) := \G_{L,N,\rho}(o,x)
\end{equation*}
for any $x \in V_L$. 
The proof of Theorem \ref{theo:theo2}  relies on the following theorem.

\begin{theorem}[{\cite[Theorem 5.1]{T}}]
For any   $d, N \in \mathbb{N}$, $L \in 2 \mathbb{N}$,
$\rho \in [0, \infty)$, 
we have that, 
\begin{equation}\label{eq:infrared}
\begin{aligned}
 \sum\limits_{\substack{x \in V_L }} \frac{\mathbb{G}_{L, N, \rho}(x) }{|V_L|} 
& \geq  \,  \frac{1}{2} \, \bigg(\mathbb{G}_{L, N, \rho}  (  \, \boldsymbol{e}_1  \,   ) \,  -  \,  \mathcal{I}_L(d)  \,   + \frac{2}{|V_L|} \, \sum\limits_{ \substack{ x  \in V^e_L \,  : \\ x_2 = \ldots  = x_d = 0}} \, \Upsilon_L(x) \mathbb{G}_{L, N, \rho}(x)\bigg), \\
\end{aligned}
\end{equation}
where $\big (\mathcal{I}_L(d) \big )_{L \in \mathbb{N}}$ is a  sequence of real numbers whose limit $L \rightarrow \infty$ exists and satisfies
\begin{equation}\label{eq:limitofIL}
\lim_{ L \rightarrow \infty   } \mathcal{I}_L(d) = \frac{r_d}{4d},
\end{equation}
and 
 $(\Upsilon_L)_{L \in \mathbb{N}}$  is a  sequence of  real-valued functions, defined by
\begin{equation*}
\begin{aligned}
\forall x \in \mathbb{Z}^d \quad  \quad  \Upsilon_L(x) &  : = \text{Re}\bigg(
\sum\limits_{ \substack{k \in V_L^*: \\ k_1 \in (-\frac{\pi}{2},\frac{\pi}{2}]}  }     e^{ -i k \cdot ( x - \boldsymbol{e}_1)}\bigg).
\end{aligned}
\end{equation*}
Here, $V_L^*:= \{ \frac{2\pi}{L} x: \, x \in V_L\}$ denotes the vertex set of the Fourier dual torus of $\T_L$. 
\end{theorem}

We will obtain \eqref{eq:permpositivityaverage} by deriving a lower bound for the sum of the first and third term appearing on the right-hand side of \eqref{eq:infrared}.

\begin{proof}[Proof of Theorem \ref{theo:theo2}]
We will show that 
\begin{equation} \label{eq:termpositive}
\begin{aligned}
& \liminf\limits_{ \substack{L \rightarrow \infty \\ \mbox{ \footnotesize $L$ even } } } \, \, \bigg[ \G_{L,N,\rho}(\boldsymbol{e}_1)+ \frac{2}{|V_L|} \, \sum\limits_{ \substack{ x  \in V^e_L \,  : \\ x_2 = \ldots  = x_d = 0}} \, \Upsilon_L(x) \mathbb{G}_{L, N, \rho}(x) \bigg] \\
& \qquad \qquad \qquad \geq   \frac{3\pi-4}{2\pi} \, \, \liminf\limits_{ \substack{L \rightarrow \infty \\ \mbox{ \footnotesize $L$ even } } } \, \mathbb{G}_{L, N, \rho}  (  \, \boldsymbol{e}_1  \,   )  + \frac{2}{\pi} \, \liminf\limits_{ \substack{L \rightarrow \infty \\ \mbox{ \footnotesize $L$ even } } } \, \,  \mathbb{G}_{L, N, \rho}(o).
\end{aligned}
\end{equation}

\vspace{-1pt}
Suppose that $L \in 2\N$. Fix $x \in V_L$ such that $x_1 \in 2\N$ and $x_2=\dots=x_d=0$. Some basic calculations, which are given in the appendix,  show that
\begin{equation} \label{eq:sumupsilon}
\Upsilon_L(x) = \begin{cases}
- L^{d-1} \, \cos\big(\frac{\pi \, x_1}{2} \big) \, \cot\big(\frac{\pi}{L} \, (x_1-1)\big) & \text{ if } L \in 4\N, \\
-L^{d-1} \, \cos\big(\frac{\pi \, x_1}{2} \big) \, \csc\big(\frac{\pi}{L} \, (x_1-1)\big) & \text{ if } L \in 2\N \setminus 4\N.
\end{cases}
\end{equation}

We will derive \eqref{eq:termpositive} for $L \in 4\N$, namely $L=4m$ for some $m \in \N$. The case $L \in 2\N \setminus 4\N$ is then similar. 
By \eqref{eq:sumupsilon} we have that, 

\begin{equation} \label{eq:calculation}
\begin{aligned}
& \frac{1}{|V_L|} \, 
\sum\limits_{ \substack{ x  \in V_L^e \,  : 
\\ x_2 = \ldots  = x_d = 0}} \, \Upsilon_L(x) \mathbb{G}_{L, N, \rho}(x) \\
& =  - \frac{1}{4m} \sum\limits_{n=-m+1}^{m}  \, \cos(\pi \,n) \, \cot\big(\frac{\pi}{4m} \, (2n-1)\big) \, \mathbb{G}_{L, N, \rho}\big((2n,0,\dots,0)\big) \\ 
& = \frac{1}{4m} \sum\limits_{\substack{n=1 \\ \text{n odd}}}^{m-1} \mathbb{G}_{L, N, \rho}\big((2n,0,\dots,0)\big) \, \Big( \cot\big(\frac{\pi}{4m} \, (2n-1)\big) - \cot\big(\frac{\pi}{4m} \, (2n+1)\big) \Big) \\
& \qquad - \frac{1}{4m} \sum\limits_{\substack{n=1 \\ \text{n even}}}^{m-1} \mathbb{G}_{L, N, \rho}\big((2n,0,\dots,0)\big) \, \Big( \cot\big(\frac{\pi}{4m} \, (2n-1)\big) - \cot\big(\frac{\pi}{4m} \, (2n+1)\big) \Big) \\
& \qquad + \frac{1}{4m} \mathbb{G}_{L, N, \rho}\big((0,0,\dots,0)\big) \, \cot(\frac{\pi}{4m})- \frac{1}{4m} \cos(m \, \pi) \cot(\frac{\pi}{2}- \frac{\pi}{4m}) \, \mathbb{G}_{L, N, \rho}\big((2m,0,\dots,0)\big),
\end{aligned}
\end{equation}
where we used that $\cot(-x)=-\cot(x)$ for any $x \in \R$ and that $\G_{L,N,\rho}(x)=\G_{L,N,\rho}(-x)$ for any $x \in V_L$.

By \cite[Theorem 2.1]{L-T-quantum} and \cite[Theorem 2.3]{T} it holds that $\mathbb{G}_{L, N, \rho}(x) \leq \mathbb{G}_{L, N, \rho}(\boldsymbol{e}_1) \leq \frac{1}{dN}$ for any $x \in V_L$. The last term of \eqref{eq:calculation} thus vanishes to zero in the limit $m \to \infty$ since $ \cot(\frac{\pi}{2})=0$. The first term is non-negative since $\cot(x)$ is monotonically decreasing on the interval $(0,\pi)$. From these considerations and using that $\lim\limits_{m \to \infty} \frac{1}{4m} \cot(\frac{\pi}{4m})=\frac{1}{\pi}$, we obtain from \eqref{eq:calculation} that
\begin{equation} \label{eq:lowerbound}
\begin{aligned}
& \liminf_{\substack{L \to \infty \\ L \in 4\N}} \, \bigg[ \G_{L,N,\rho}(\boldsymbol{e}_1)+ \frac{2}{|V_L|} \, 
\sum\limits_{ \substack{ x  \in V_L^e \,  : 
\\ x_2 = \ldots  = x_d = 0}} \, \Upsilon_L(x) \, \mathbb{G}_{L, N, \rho}(x) \bigg]\\
& \geq \liminf_{m \to \infty} \big(\, \mathbb{G}_{4m, N, \rho}(\boldsymbol{e}_1) \,  a_m \big) + \frac{2}{\pi} \, \liminf_{\substack{L \to \infty \\ L \in 4\N}} \, \mathbb{G}_{L, N, \rho}(o),
\end{aligned}
\end{equation}
where for $m \in \N$, 
\[
a_m:= 1-\frac{2}{4m} \sum\limits_{n=1 }^{\lfloor \frac{m-1}{2} \rfloor}  \Big( \cot\big(\frac{\pi}{4m} \, (4n-1)\big) - \cot\big(\frac{\pi}{4m} \, (4n+1)\big) \Big).
\]
For $ 0 < |x| < 1$, expansion into Taylor series gives that $\cot(\pi \, x) = \frac{1}{\pi \, x} + \frac{1}{\pi} \, \sum_{l=1}^\infty c_l \, x^{2l-1}$, where $c_l := (-1)^l \, \frac{(2\pi)^{2l} \, B_{2l}}{(2l)!}$ for any $l \in \N$. Here, $(B_n)_{n \geq 0}$ denote the Bernoulli numbers, see e.g. \cite[p. 220]{Remmert}.
Using the binomial theorem and the Leibniz formula for $\pi$ we thus deduce that,
\begin{equation} \label{eq:Laurent}
\begin{aligned}
a_m
& = 1-\frac{2}{\pi} \, \sum\limits_{n=1 }^{\lfloor \frac{m-1}{2} \rfloor}  \Big( \frac{1}{4n-1}-\frac{1}{4n+1}\Big) + \frac{4}{\pi} \, \sum_{l=1}^\infty f_m(l) \\
& \geq  1-\frac{2}{\pi} \, \sum_{n=1}^\infty \frac{(-1)^{n+1}}{2n+1} - \frac{4}{\pi} \, \sum_{l=1}^\infty |f_m(l)| = \frac{3\pi-4}{2\pi} - \frac{4}{\pi} \, \sum_{l=1}^\infty |f_m(l)|,
\end{aligned}
\end{equation}
where for any $m,l \in \N$, 
\[
f_m(l) := c_{l} \, \Big(\frac{1}{4m}\Big)^{2l} \, \sum_{\substack{k=0: \\ k \text{ odd}}}^{2l-1} \binom{2l-1}{k} \, \sum\limits_{n=1 }^{\lfloor \frac{m-1}{2} \rfloor} (4n)^{2l-1-k}.
\]
We will now show that 
\begin{equation} \label{eq:limitf1}
\lim_{m \to \infty} \sum_{l=1}^\infty |f_m(l)|=0.
\end{equation}
For the next calculation we use the upper bound $|B_{2l}| < 4 \, \frac{(2l)!}{(2\pi)^{2l}}$ for any $l \in \N$, see e.g. \cite[p. 332]{Remmert}, which implies that $|c_l| < 4$ for any $l \in \N$. We then have that,
\begin{equation} \label{eq:upperboundsum}
\begin{aligned}
\sum_{l=1}^\infty |f_m(l)| 
& \leq 4 \, \sum_{l=1}^\infty \Big(\frac{1}{4m}\Big)^{2l} \, \frac{m-1}{2} \,  \, \sum_{\substack{k=0: \\ k \text{ odd}}}^{2l-1} \binom{2l-1}{k} \, (2m-2)^{2l-1-k} \\
& \leq \frac{1}{4} \, \sum_{l=1}^\infty \Big(\frac{1}{4m}\Big)^{2l-1} \, \Big[ (2m-1)^{2l-1}-(2m-3)^{2l-1} \Big],  
\end{aligned}
\end{equation}
where in the last step we used that $\sum_{k=0}^n \binom{n}{k} x^{n-k}  \, \mathbbm{1}_{\{k \text{ odd}\}}= \frac{1}{2} \big[(x+1)^n-(x-1)^n\big]$ for any $n \in \N$ and $x \in \R$.
Now for any $m \in \N$,
\begin{equation*}
\Big(\frac{1}{2}-\frac{1}{4m}\Big)^{2l-1}-\Big(\frac{1}{2}-\frac{3}{4m}\Big)^{2l-1}  \leq \Big(\frac{1}{2}\Big)^{2l-1},
\end{equation*}
and $\sum_{l=1}^\infty \big(\frac{1}{2}\big)^{2l-1} < \infty$. We can thus apply reverse Fatou's Lemma and obtain from \eqref{eq:upperboundsum} that
\begin{equation} \label{eq:limitf}
0 \leq \limsup_{m \to \infty} \sum_{l=1}^\infty |f_m(l)| \leq \frac{1}{4} \, \sum_{l=1}^\infty \limsup_{m \to \infty} \bigg[ \Big(\frac{1}{2}-\frac{1}{4m}\Big)^{2l-1}-\Big(\frac{1}{2}-\frac{3}{4m}\Big)^{2l-1} \bigg]=0.
\end{equation}
From \eqref{eq:limitf} we deduce \eqref{eq:limitf1}. In particular, $a_m \geq 0$ for $m$ large enough. These considerations together with \eqref{eq:Laurent} imply that
\begin{equation} \label{eq:am}
\begin{aligned}
& \liminf_{m \to \infty} \,  \big(\, \mathbb{G}_{4m, N, \rho}(\boldsymbol{e}_1) \,  a_m \big) 
& \geq \frac{3\pi-4}{2\pi} \, \liminf_{m \to \infty} \, \mathbb{G}_{4m, N, \rho}(\boldsymbol{e}_1) \, \, .
\end{aligned}
\end{equation}

From \eqref{eq:lowerbound} and \eqref{eq:am} we deduce \eqref{eq:termpositive}.
It remains to provide a lower bound on $\lim_{L \to \infty}\G_{L,N,\rho}(\boldsymbol{e}_1)$. Let $L \in 2\N$, $L>2$. Let $\Omega^\ell$ denote the set of triples $\pi=(M,d_1,d_2)$ such that $M \subset V_L$ and $d_1,d_2 \in D_{\T_L}(M)$. Each triplet can be seen as a system of monomers and disjoint self-avoiding loops. For any $\pi \in \Omega^\ell$, we let $\tilde{\mathcal{L}}(\pi)$ be the number of loops in $\pi$. We define a probability measure on $\Omega^\ell$,
\begin{equation*}
\forall \pi \in \Omega^\ell \qquad \tilde{\P}_{L,N,\rho}(\pi) := \frac{\rho^{|M(\pi)|} \big(\frac{N}{2})^{\tilde{\mathcal{L}}(\pi)}}{\tilde{Z}_{L,N,\rho}},
\end{equation*}
where $\tilde{Z}_{L,N,\rho}$ is a normalization constant. We have the equality (see e.g. \cite{T}), 
\begin{equation} \label{eq:G(e1)bound}
\mathbb{G}_{L, N, \rho}(\boldsymbol{e}_1) = \frac{1}{d \, N } \,  \big(1-\tilde{\P}_{L,N,\rho}(o \text{ is a monomer})\big).
\end{equation}
The probability appearing on the right-hand side of \eqref{eq:G(e1)bound} can be reformulated as a probability in the \textit{random path model} for special choices of the parameters, see e.g. \cite{L-T-first,T}. This reformulation allows the application of the Chessboard Estimate, see e.g. \cite[Proposition 4.5]{T}, namely
\begin{equation} \label{eq:probmonomerbound}
\tilde{\P}_{L,N,\rho}(o \text{ is a monomer}) \leq \tilde{\P}_{L,N,\rho}(\forall x \in V_L, \, x \text{ is a monomer})^{\frac{1}{|V_L|}} \leq \rho,
\end{equation}
where we used that $\tilde{Z}_{L,N,\rho} \geq 1$. Together with \eqref{eq:G(e1)bound} the upper bound \eqref{eq:probmonomerbound} implies that
\begin{equation} \label{eq:Chessboardapplication}
\mathbb{G}_{L, N, \rho}(\boldsymbol{e}_1)  \geq \frac{1}{dN}(1-\rho).
\end{equation}
Combining \eqref{eq:infrared}, \eqref{eq:limitofIL}, \eqref{eq:termpositive} and \eqref{eq:Chessboardapplication} concludes the proof of \eqref{eq:permpositivityaverage}.  Note that the second term on the right-hand side of \eqref{eq:termpositive} is non-negative and is neglected.  From \eqref{eq:permpositivityaverage} and \cite[Theorem 2.2]{L-T-quantum} we further deduce \eqref{eq:eq2theo2}.
For the last statement, equation \eqref{eq:eq3theo2}, note that \eqref{eq:permpositivityaverage} implies the existence of $\tilde{C}<C$ such that for any $\alpha \in (0,1)$ and for any $L \in 2\N$, 
\begin{equation*}
\mathbb{P}_{L, N, \rho} \big (  \, | X  |_1  \, \leq \,  \alpha \, L  \,  \big ) = \frac{\sum\limits_{x \in V_{\lfloor \alpha \, L \rfloor}}\G_{L,N,\rho}(o,x)}{\sum\limits_{x \in V_L} \G_{L,N,\rho}(o,x)} \leq \frac{\alpha^d}{\tilde{C}},
\end{equation*}
where we used that $\G_{L,N,\rho}(x) \leq 1$ for all $x \in V_L$ (see \cite[Theorem 2.3]{T})  in the last step. Choosing $\alpha>0$ small enough and setting $c < \min(\alpha, 1-\frac{\alpha^d}{\tilde{C}})$ proves \eqref{eq:eq3theo2} and concludes the proof of the theorem.
\end{proof}

\appendix

\section{Appendix}
\label{sect:appendix}

\begin{proof}[Proof of \eqref{eq:sumupsilon} in the proof of Theorem \ref{theo:theo2}]
Let $L \in 2\N$ and $x \in V_L$ such that $x_1 \in 2\N$ and $x_2=\dots=x_d=0$. 
We have that
\begin{equation} \label{eq:calculusgeneral}
\Upsilon_L(x) = \begin{cases}
L^{d-1} \, \text{Re}\bigg(
\sum\limits_{k=-\frac{L}{4}+1}^{\frac{L}{4}}    e^{ - i \, \frac{2\pi}{L} \, k \, (x_1-1)}\bigg) & \text{ if } L \in 4\N, \\
 L^{d-1} \, \text{Re}\bigg(\sum\limits_{k=- \frac{L-2}{4}+1}^{\frac{L-2}{4}}    e^{ - i \, \frac{2\pi}{L} \, k \, (x_1-1)}\bigg) & \text{ if } L \in 2\N \setminus 4\N.
\end{cases}
\end{equation}
We now derive \eqref{eq:sumupsilon} for $L \in 4 \N$. The case $L \in 2\N \setminus 4\N$ is then similar. Consider $L=4n$ for some $n \in \N$ and take $x \in \Z$ odd. Using that $\cos(\frac{\pi}{2}x) =0$,  we have that
\begin{equation} \label{eq: calculusupsilonA1}
\begin{aligned}
\text{Re}\bigg(
\sum\limits_{k=-n+1}^n    e^{ - i \, \frac{\pi}{2n} \, k \, x}\bigg) 
& = \text{Re}\bigg( e^{ i \, \frac{\pi}{2} \, x} \, e^{-i \, \frac{\pi}{2n} \, x} 
\sum\limits_{k=0}^{2n-1}    \Big(e^{ - i \frac{\pi}{2n} \, x}\Big)^k\bigg) \\
& = 2 \,\sin(\frac{\pi}{2} \, x) \, \text{Im}\bigg(e^{-i \, \frac{\pi}{2n} \, x} \, \frac{1}{e^{-i \, \frac{\pi}{2n}\, x}-1} \bigg),
\end{aligned}
\end{equation}
where in the last step we used that $\cos(\pi \, x)=-1$ and $\sin(\pi \, x)=0$.
We have that
\begin{equation} \label{eq:calculusupsilonA2}
\begin{aligned}
\text{Im}\bigg(e^{-i \, \frac{\pi}{2n} \, x} \, \frac{1}{e^{-i \, \frac{\pi}{2n}\, x}-1} \bigg)
 = \cos(\frac{\pi}{2n} \, x) \, \text{Im}\bigg(\frac{1}{e^{-i \, \frac{\pi}{2n}\, x}-1} \bigg) - \sin(\frac{\pi}{2n} \, x) \, \text{Re}\bigg(\frac{1}{e^{-i \, \frac{\pi}{2n}\, x}-1} \bigg).
\end{aligned}
\end{equation}
Solving $\frac{1}{e^{-i \, \frac{\pi}{2n}\, x}-1} = a + ib$ for $a$ and $b$ gives $a = -\frac{1}{2}$ and $b= \frac{\sin(\frac{\pi}{2n} \, x )}{2-2 \cos(\frac{\pi}{2n} x )}$. From  \eqref{eq:calculusupsilonA2}, we thus obtain  that
\begin{equation} \label{eq:calculuslast}
\text{Im}\bigg(e^{-i \, \frac{\pi}{2n} \, x} \, \frac{1}{e^{-i \, \frac{\pi}{2n}\, x}-1} \bigg)
= \frac{1}{2} \, \frac{\sin(\frac{\pi}{2n} \, x)}{1-\cos(\frac{\pi}{2n} \, x)} 
 = \frac{1}{2} \, \cot(\frac{\pi}{4n} \, x).
\end{equation}
From \eqref{eq:calculusgeneral}, \eqref{eq: calculusupsilonA1} and \eqref{eq:calculuslast} we deduce \eqref{eq:sumupsilon}.
\end{proof}

\section*{Acknowledgements} The authors thank the German Research Foundation (project number 444084038, priority program SPP2265) for financial support. AQ additionally thanks the German Research Foundation through IRTG 2544 and the German Academic Exchange Service (grant number 57556281) for financial support. 
The authors also thank the editor and the two anonymous  referees for  carefully reviewing the paper. 


\end{document}